\documentclass[12pt]{amsart}
\usepackage{amssymb}
\usepackage{color}
\newtheorem{claim}{}[section]
\newtheorem{theorem}[claim]{Theorem}

\newtheorem{proposition}[claim]{Proposition}

\renewenvironment{proof}{\noindent{\it Proof. \hskip0pt}}
                      {$\square$\par\medskip}

\begin{document}
\baselineskip 6.2 truemm
\parindent 1.5 true pc

\newcommand\lan{\langle}
\newcommand\ran{\rangle}
\newcommand\tr{{\text{\rm Tr}}\,}
\newcommand\ot{\otimes}
\newcommand\wt{\widetilde}
\newcommand\join{\vee}
\newcommand\meet{\wedge}
\renewcommand\ker{{\text{\rm Ker}}\,}
\newcommand\im{{\text{\rm Im}}\,}
\newcommand\mc{\mathcal}
\newcommand\transpose{{\text{\rm t}}}
\newcommand\FP{{\mathcal F}({\mathcal P}_n)}
\newcommand\ol{\overline}
\newcommand\JF{{\mathcal J}_{\mathcal F}}
\newcommand\FPtwo{{\mathcal F}({\mathcal P}_2)}
\newcommand\hada{\circledcirc}
\newcommand\id{{\text{\rm id}}}
\newcommand\tp{{\text{\rm tp}}}
\newcommand\pr{\prime}
\newcommand\e{\epsilon}
\newcommand\inte{{\text{\rm int}}\,}
\newcommand\ttt{{\text{\rm t}}}
\newcommand\spa{{\text{\rm span}}\,}
\newcommand\conv{{\text{\rm conv}}\,}
\newcommand\rank{\ {\text{\rm rank of}}\ }
\newcommand\vvv{\mathbb V_{m\meet n}\cap\mathbb V^{m\meet n}}
\newcommand\ppp{\mathbb P_{m\meet n} + \mathbb P^{m\meet n}}
\newcommand\re{{\text{\rm Re}}\,}
\newcommand\la{\lambda}
\newcommand\msp{\hskip 2pt}
\newcommand\ppt{\mathbb T}
\newcommand\rk{{\text{\rm rank}}\,}
\def\cc{\mathbb{C} }
\def\pp{\mathbb{P} }
\def\rr{\mathbb{R} }
\def\qq{\mathbb{Q} }
\def\a{\alpha }
\def\b{\beta }

\title{Entanglement witnesses arising from exposed positive linear maps}

\author{Kil-Chan Ha}
\address{Faculty of Mathematics and Statistics, Sejong University, Seoul 143-747, Korea}
\email{kcha@sejong.ac.kr}

\author{Seung-Hyeok Kye}
\address{Department of Mathematics and Institute of Mathematics\\Seoul National University\\Seoul 151-742, Korea}
\email{kye@snu.ac.kr}

\thanks{KCH is partially supported by NRFK 2011-0006561. SHK is partially supported by NRFK 2011-0001250}



\begin{abstract}
We consider entanglement witnesses arising from positive linear maps
which generate exposed extremal rays. We show that every
entanglement can be detected by one of these witnesses, and this
witness detects a unique set of entanglement among those.
Therefore, they provide a minimal set of witnesses
to detect all entanglement in a sense.
Furthermore, if those maps are indecomposable then they detect
large classes of entanglement with positive partial
transposes which have nonempty relative interiors in the cone
generated by all PPT states. We also provide a one parameter
family of indecomposable positive linear maps which generate
exposed extremal rays. This gives the first examples of such maps
between three dimensional matrix algebra.
\end{abstract}

\maketitle

\section{Introduction}

The notion of entanglement is one of the key current research areas of quantum physics
in the relation of possible applications to quantum computation and quantum information theories.
Entanglement is a positive semi-definite $nm\times nm$ matrix in $M_{nm}=M_n\otimes M_m$ which is not separable,
where $M_n$ denotes the $C^*$-algebra of all $n\times n$ matrices over the complex field.
Recall that a positive semi-definite matrix gives rise to a state on the matrix algebra through the
Hadamard product if it is normalized.
A positive semi-definite matrix in $M_n \otimes M_m$
is said to be separable if it is the convex sum of rank one positive semi-definite
matrices onto product vectors of the form $x\otimes y\in \mathbb C^n\otimes \mathbb C^m$.
We denote by $\mathbb V_1$ the cone of all separable ones. It is easy to see that
the convex cone $\mathbb V_1$ coincides with the tensor product $M_n^+\otimes M_m^+$ of the positive cones
consisting of all positive semi-definite matrices.
Therefore, entanglement consists of $(M_n\otimes M_m)^+\setminus M_n^+\otimes M_m^+$.

If ${\mathcal A}$ and $\mathcal B$ are commutative $C^*$-algebras then the two cones
$({\mathcal A}\otimes{\mathcal B})^+$ and ${\mathcal A}^+\otimes{\mathcal B}^+$ coincide,
and so the notion of entanglement reflects non-commutative order structures in nature.

The most general method to determine if a given state is separable or not is to use the duality
between the space $M_n\otimes M_m$ and the space ${\mathcal L}(M_m,M_n)$ of all linear maps from $M_m$ into $M_n$.
The bi-pairing between these spaces is given by
\begin{equation}\label{pairing}
\langle y\otimes x, \phi \rangle ={\text{\rm Tr}}(\phi(x)y^\ttt)
\end{equation}
for $x\in M_m, y\in M_n$ and $\phi\in{\mathcal L}(M_m,M_n)$.
It is possible to define this pairing in more general context, and used \cite{arveson} to define
topologies in the space of linear maps between operator algebras, which are useful to prove
Hahn-Banach type extension theorems for completely positive linear maps. This pairing had been also used
by Woronowicz \cite{woronowicz} to show that there exists an indecomposable positive linear map
from $M_2$ into $M_4$.
It had been also used in \cite{stormer-dual} to study extendibility of positive linear maps on $C^*$-algebras.

From now on, every vector in $\mathbb C^m$ may be considered as an $m\times 1$ matrix,
and we denote by $\bar x$ and $x^*$ the complex conjugate and the Hermitian conjugate of $x$. In this notation, rank one matrix $|x\rangle \langle y|$ is written by $xy^*$.
A linear map between $C^*$-algebras is said to be positive if it sends positive elements
into themselves. We denote by $\mathbb P_1$ the convex cone of all positive linear maps
in ${\mathcal L}(M_m,M_n)$. It turns out \cite{eom-kye} that the cones $\mathbb V_1$ and $\mathbb P_1$ are dual each other
with respect to the pairing (\ref{pairing}) in the following sense:
\begin{equation}\label{duality}
\begin{aligned}
A\in\mathbb V_1\ &\Longleftrightarrow\ \langle A,\phi\rangle \ge 0\ {\text{\rm for each}}\ \phi\in\mathbb P_1,\\
\phi\in\mathbb P_1\ &\Longleftrightarrow\ \langle A,\phi\rangle \ge 0\ {\text{\rm for each}}\ A\in\mathbb V_1.
\end{aligned}
\end{equation}
Therefore, we see that a state $A$ is entangled if and only if
there exists a positive linear map such that $\lan A,\phi\ran <0$.

For a positive linear map $\phi\in \mathcal L (M_m , M_n)$, define the
matrix $W_{\phi} \in M_m\otimes
M_n$ by
\[
W_{\phi}=(\text{id}\otimes \phi)P^+=\sum_{ij=1}^m e_{ij}\otimes \phi(e_{ij}),
\]
where $\{e_{ij}:i,j=1,2,\ldots,m\}$ denotes the usual matrix units in $M_m$ and $P^+$
denotes the projector onto the maximally entangled state $\sum_{i=1}^m e_i\otimes e_i$ in $\mathbb C^m\otimes \mathbb C^m$.
Since we have the relation
\begin{equation}\label{eq:JC}
\langle x\otimes y|W_\phi|x\otimes
y\rangle=\text{Tr}(W_{\phi}|x\otimes y\rangle\langle x\otimes y|)
= y^* \phi (\bar x\bar x^*) y =\langle \bar y\bar y^* \otimes \bar
x\bar x^*, \phi \rangle,
\end{equation}
and $\mathbb V_1=M_n^+\otimes M_m^+$ is spanned by $\bar y\bar y^* \otimes \bar x\bar x^*$, we see that
$A_0$ is an entangled state if and only if there is a hermitian matrix $W$ with the property:
\begin{equation}\label{def-wit}
{\text{\rm Tr}}\, (WA_0)<0,\qquad {\text{\rm Tr}}\, (WA)\ge 0\ {\text{\rm for each}}\ A\in \mathbb V_1.
\end{equation}
In this sense, the duality
(\ref{duality}) is equivalent to the separability criterion given in \cite{horo-1} under the
Jamio\l kowski-Choi isomorphism  \cite{choi75-10} $\phi\mapsto W_\phi$. Note that the inverse $W\mapsto \phi$ is given by
$$
\phi (X)=\sum_{i,j=1}^m x_{ij}W_{ij},
$$
for $X=[x_{ij}]\in M_m$ and $W=\sum_{i,j}e_{ij}\otimes W_{ij}\in M_m\otimes M_n$.

A Hermitian matrix $W$ is said to be an entanglement witness if there is an entanglement $A_0$ with the property
(\ref{def-wit}). Therefore, any entanglement witness is of the form $W_\phi$ for a positive map $\phi$ which is not
completely positive. In the language of the duality (\ref{duality}), we see that every positive map $\phi$ detects
an entangled state $A$ in the sense of $\lan A,\phi\ran <0$ whenever $\phi$ is not completely positive.
Since our presentation heavily depends on the duality (\ref{duality}),
we will call sometimes a positive map itself an entanglement
witness if it detect entanglement. After Terhal \cite{terhal} introduced the terminology of entanglement
witness, Lewenstein, Kraus, Cirac and Horodecki \cite{lew00} studied the notion of optimal
entanglement witness which detects a maximal set of entanglement, and addressed \cite{lew00_1} a fundamental
question to find a minimal set of witnesses to detect all entanglement.

In this note, we consider the entanglement witnesses arising from positive linear maps
which generate exposed extremal rays of the cone $\mathbb P_1$.
A positive linear map will be said to be just {\sl exposed} if it generates an exposed ray of the cone $\mathbb P_1$.
Note that an exposed positive map automatically generates an extremal ray.
We denote by $\mathbb W$ the set of all
those witnesses arising from exposed positive maps. We show that any entanglement can be detected
by a witness in the family $\mathbb W$. Furthermore, different witnesses in $\mathbb W$ detect
different sets of entanglement.  Therefore, the family $\mathbb W$
provides a minimal set of witnesses in a sense. It should be noted that any dense subset of $\mathbb W$ also
detects all entanglement. Recently, there are researches \cite{gan}, \cite{hou} on the witnesses
detecting the common set of entangled states.

In 1980, Choi \cite{choi-ppt} observed that if $A\in (M_n\otimes M_m)^+$ belongs to
$M_n^+\otimes M_m^+$ then the partial transpose $A^\tau$ of $A$ is again positive semi-definite. An equivalent formulation
was given in \cite{peres}. This gives us a simple necessary condition for separability, and called the PPT
criterion for separability. We denote by $\mathbb T$ the convex cone of all positive semi-definite matrices in $(M_n\otimes M_m)^+$
whose partial transposes are also positive semi-definite. The PPT criterion tells us the relation $\mathbb V_1\subset \mathbb T$.
When $m=2$, Woronowicz \cite{woronowicz} show that $\mathbb V_1= \mathbb T$ if and only if $n\le 3$, and
exhibited an explicit example in  $\mathbb T\setminus\mathbb V_1$ for the case of $m=2$ and $n=4$.
This kind of example is called a PPT entangled state if
it is normalized. The first example of PPT entangled state in the case of $m=n=3$ was given in \cite{choi-ppt}.

The most important examples of positive linear maps come from elementary operators of the forms
$$
\phi_V: X\mapsto V^*XV,
$$
for $m\times n$ matrices $V$. The convex sum of these maps are said to be a completely positive linear map.
After the notion of completely positive linear maps arises \cite{stine} in the context of the representation theory of $C^*$-algebras
in the fifties, they played key roles in several research areas of the theory of operator algebras, notably in the characterization
of nuclear $C^*$-algebras during the seventies.

The transpose map which assigns the transpose $X^\ttt$ to a given matrix $X$ is a typical example of a positive map
which is not completely positive. The convex sum of the following maps
$$
\phi^V: X\mapsto V^*X^\ttt V
$$
with $m\times n$ matrices $V$ is said to be completely copositive, and the convex sum of a completely
positive map and a completely copositive map is called a decomposable positive linear map.
We denote by $\mathbb D$ the convex cone of all decomposable positive linear maps.
The first example of an indecomposable positive linear map was given by Choi \cite{choi75}, in the case of
$m=n=3$. It was also shown by
Woronowicz \cite{woronowicz} that every positive linear map from $M_2$ into $M_n$ is decomposable if and only if
$n\le 3$ by the duality. The first explicit example of an indecomposable positive
linear map from $M_2$ into $M_4$ was given in terms of non-jordanian type in \cite{woronowicz-1}. See also \cite{tang} for more such examples.

It was also shown in \cite{eom-kye} that the cones $\mathbb T$ and $\mathbb D$ are dual each other
in the same sense as in (\ref{duality}). Therefore, in order to detect PPT entangled states lying in $\mathbb T\setminus\mathbb V_1$,
we need indecomposable positive linear maps in $\mathbb P_1\setminus\mathbb D$. This explains
the importance of indecomposable positive linear maps as detectors of PPT entangled states.
We show that witness arising from an exposed indecomposable positive map detects a set of PPT entangled states
which has a nonempty relative interior in the cone $\mathbb T$.

It was shown in \cite{yopp} and \cite{marcin_exp} that the map $\phi_V$ is exposed for each matrix $V$.
It is now clear that the maps $\phi_V$ and $\phi^V$ exhaust all decomposable positive linear maps
which generate exposed rays of $\mathbb P_1$.
In spite of the importance of exposed indecomposable positive maps, very few such maps seem to be known
to the specialists. For example,
Woronowicz \cite{woro_letter}  kindly showed the authors that the example of a non-extendible positive map from $M_2$ into $M_4$
in \cite{woronowicz-1} is actually exposed.
To the best knowledge of the authors, there is no known example of such a map between $M_3$. We provide
a one parameter family of such maps from $M_3$ into $M_3$, by showing that
some maps in \cite{cho-kye-lee} are exposed indecomposable positive linear maps.

In the next section, we consider the notion of duality for convex cones in a general framework, and apply
that to show the above mentioned result for witnesses arising from exposed positive linear maps. After
we consider the witnesses arising from exposed indecomposable positive linear maps in Section 3, we provide
the above mentioned examples in the last section.

The authors are grateful to Stanis\l aw Lech Woronowicz and Dariusz Chru\'{s}ci\'{n}ski
for the correspondences with them which were very useful to prepare this paper.

\section{Convex cones and their duals}

Let $X$ and $Y$ be finite dimensional normed spaces, which are dual each other
with respect to a bilinear pairing $\lan\ , \ \ran$.
For a subset $C$ of $X$, we define the {\sl dual cone} $C^\circ$ by
$$
C^\circ=\{y\in Y: \lan x,y\ran \ge 0\  {\text{\rm for each}}\  x\in C\},
$$
and the dual cone $D^\circ\subset X$ similarly for a subset $D$ of $Y$.
We assume that the pairing is non-degenerate on the convex cone $C$ of $X$ in the following sense:
$$
x\in C,\ \lan x,y\ran =0\ {\text{\rm for each}}\ y\in C^\circ\
\Longrightarrow\ x=0.
$$
For a face $F$ of a closed convex cone $C$ of $X$, we define the subset $F^\pr$
of $C^\circ$ by
$$
F^\pr=\{ y\in C^\circ:\lan x,y\ran =0\ {\text{\rm for each}}\ x\in F\}.
$$
It is then clear that $F^\pr$ is a closed face of $C^\circ$.
It is also clear that $F\subset F^{\pr\pr}$ for any face $F$ of $C$.
Note that a closed face $F$ of a closed convex cone $C$ is exposed
with respect to the pairing $\lan\ ,\ \ran$ if and only if the equality $F=F^{\pr\pr}$ holds.
We say that a point $x$ in a convex cone $C$ is {\sl exposed} if it generates
an exposed ray of $C$. It is clear that an exposed point of a cone generates an extremal ray.

It was shown in \cite{eom-kye} that
if $F$ is a maximal face of $C^\circ$ then there is an extremal ray
$L$ of $C$ such that $F=L^\pr$. Especially, every maximal face is exposed.
Conversely, if $L$ is an exposed face of $C$ which is minimal among nonzero exposed faces,
then $L^\pr$ is a maximal face of $C^\circ$. Especially, if $L$ is an exposed ray of $C$ then
$L^\prime$ is a maximal face of $C^\circ$.

Let $C$ be a closed convex cone of $X$.
If $y\in Y\setminus C^\circ$ then there exists $x\in C$ such that $\lan x,y\ran <0$ by the definition
of the dual cone. In this case, we say that $x\in C$ {\sl detects} the element $y\in Y\setminus C^\circ$.
Since every element of the cone $C$ is the convex sum of points of $C$ generating extremal rays, it is apparent to see
that $y\in C^\circ$ if and only if
$\lan x,y\ran \ge 0$ for every $x\in C$ generating an extremal ray. Since every extremal ray is the limit of exposed rays
by Straszewicz's Theorem (see Theorem 18.6 of \cite{rock}), we also see that
$y\in C^\circ$ if and only if $\lan x,y\ran \ge 0$ for every exposed point $x\in C$.
In other words, every element of $y\in Y\setminus C^\circ$ is detected by an exposed point of $C$.
We state this as follows:

\begin{proposition}\label{exposed-det}
Let $C$ be a closed convex cone of $X$. For $y\in Y$, the following are equivalent:
\begin{enumerate}
\item[(i)]
$y\notin C^\circ$.
\item[(ii)]
There exists an exposed point $x\in C$ such that $\lan x,y\ran < 0$.
\end{enumerate}
\end{proposition}

By Proposition \ref{exposed-det}, we see that a positive semi-definite matrix $A\in M_n\otimes M_m$ is
entangled if and only if there exists an exposed positive linear map such that
$\lan A,\phi\ran <0$. In this sense, we see that every entanglement is detected by
an exposed positive linear map.

We denote by $E_\phi$ the set of all entanglement detected by $\phi\in\mathbb P_1$.
Let $\phi,\psi\in \mathbb W$. If $E_\phi$ and $E_\psi$ coincide then the maximal faces ${L_\phi}^\prime$ and ${L_\psi}^\prime$
also coincide, where $L_\phi$ denotes the ray generated by $\phi$. Note that $L_\phi$ itself is an exposed face by
definition. Therefore, we have
$$
L_\phi={L_\phi}^{\prime\prime}={L_\psi}^{\prime\prime}=L_\psi.
$$
This means that $\phi$ is a scalar multiple of $\psi$.
Therefore, we conclude that if $\phi,\psi\in\mathbb W$ are different with the exception of scalar multiplication
then they detect different sets of entanglement.

We know all extremal rays of the cone $\mathbb V_1$, by the definition of separability.
They consist of all rank one projectors onto product vectors in $\mathbb C^n\otimes\mathbb C^m$.
This enables us to find all maximal faces of the cone $\mathbb P_1$, as was done in \cite{kye-canad}.
We will explain that for the later use.

For a projector $(\bar y\otimes x)(\bar y\otimes x)^*\in\mathbb V_1$ onto the product
vector $\bar y\otimes x\in\mathbb C^n\otimes\mathbb C^m$,
we have
$$
\begin{aligned}
\lan (\bar y\otimes x)(\bar y\otimes x)^*,\phi\ran
&=\lan \bar y \bar y^*\otimes xx^*,\phi\ran\\
&={\text{\rm Tr}}\, (\phi(xx^*)yy^*)
=( \phi(xx^*)y\, |\, y)_{\mathbb C_n},
\end{aligned}
$$
where the last expression denotes the inner product which is linear in the first variable and conjugate linear
in the second variable. Therefore, if $L$ is an extremal ray of the cone $\mathbb V_1$ determined by a product vector
$\bar y\otimes x$ then the dual face $L^\prime$ is the set of all positive linear maps satisfying the condition
$$
( \phi(xx^*)y\, |\, y)=0.
$$
We denote by $P_\phi$ the set of all product vectors $\bar y\otimes x$ satisfying the above relation.
Therefore $P_{\phi}$ determines the dual face $L_{\phi}'$ of $\mathbb V_1$
for the ray $L_{\phi}$ generated by the positive linear map $\phi$.

\section{PPT entanglement detected by exposed indecomposable maps}

In this section, we compare the boundary structures of the two cones $\mathbb V_1$ and $\mathbb T$, to distinguish
the role of entanglement witnesses arising from exposed indecomposable positive maps among all
those from exposed positive maps. To begin with,
we recall that every face of the cone $\mathbb D$ is determined \cite{kye_decom} by a pair $(D,E)$ of subspaces of the space
$M_{m\times n}$ of all $m\times n$ matrices. More precisely, every face of the cone $\mathbb D$ is of the form
$$
\sigma (D,E)=\left\{\textstyle\sum_i \phi_{V_i}+\sum_j\phi^{W_j}: V_i\in D,\ W_j\in E\right\}
$$
for a pair $(D,E)$. It should be noted that not every pair gives rise to a face of the cone $\mathbb D$.
Using the duality between two cones $\mathbb D$ and $\mathbb T$,
the authors \cite{ha_kye_04} showed that every face of $\mathbb T$ is of the form
$$
\tau (D,E)=\{A\in\mathbb T: {\mathcal R}A\subset D,\ {\mathcal R}A^\tau\subset E\},
$$
which is nothing but the dual face of the face $\sigma(D^\perp,E^\perp)$ of the cone $\mathbb D$, where ${\mathcal R}A$
denotes the range space of $A$ with the identification of the space $M_{m\times n}$
with $\mathbb C^n\otimes \mathbb C^m$, under which a product vector $\bar y\otimes x$ corresponds to a rank one matrix $xy^*$.
Especially, every face of $\mathbb T$ is exposed. Note \cite {byeon-kye} that there is an unexposed face of
$\mathbb D$, even in the simplest case $m=n=2$.

Because the cone $\mathbb V_1$ is a subcone of $\mathbb T$, the faces of $\mathbb V_1$ are naturally
divided in two categories as was studied in \cite{choi_kye}. A face of $F$ of $\mathbb V_1$ is said to be
induced by the face $\tau(D,E)$ of $\mathbb T$ if it is of the form
$$
F=\mathbb V_1\cap \tau(D,E)
$$
with the additional condition $\mathbb V_1\cap \inte\tau(D,E)\neq\emptyset$, where $\inte C$ denotes
the relative interior of the convex set $C$ with respect to the affine manifold generated by $C$.
It was also shown in \cite{choi_kye} that a pair $(D,E)$ gives rise to the face $\tau(D,E)$ of $\mathbb T$
which induces a face of $\mathbb V_1$ if and only if
the pair $(D,E)$ satisfies the range criterion, that is, there is a family of
product vector $\{y_\iota\otimes x_\iota\}$ satisfying the relation
$$
D=\spa\{y_\iota\otimes x_\iota\},\qquad E=\spa\{y_\iota\otimes \bar x_\iota\}.
$$

Now, we restrict our attention to maximal faces of the cone $\mathbb V_1$ which determine the whole boundary.
We begin with the maximal face dual to a completely copositive map $\phi^V$, which may be an optimal entanglement
witness under the Jamio\l kowski-Choi isomorphism.  In this case, we have
\begin{equation}\label{hsgg}
\begin{aligned}
\lan(\bar y\otimes x)(\bar y\otimes x)^*,\phi^V\ran &=y^*\phi^V (xx^*)y\\
&=y^*V^*\bar x\bar x^* Vy\\
&=|\bar x^* Vy|^2\\
&=|(V|\bar xy^*)|^2=|(V|\bar y\otimes \bar x)|^2
\end{aligned}
\end{equation}
from the equation~\eqref{eq:JC} and the identification between $\bar y\otimes x$ and $xy^*$.
Therefore, we see that the set $P_{\phi^V}$
consists of the partial conjugates of all product vectors orthogonal to $V$.
This relation (\ref{hsgg}) also tells us generally that for $\phi\in\mathbb P_1$ the
partial conjugates of the product vectors in
$P_\phi$ has a nontrivial orthogonal complement if and only if the double dual ${L_\phi}^{\prime\prime}$
contains a completely copositive map.
By the same calculation, we also have
$$
\lan(\bar y\otimes x)(\bar y\otimes x)^*,\phi_V\ran =|(V\, |\, \bar y\otimes  x)|^2,
$$
and so the set $P_{\phi_V}$ consist of all product vectors orthogonal to $V$.

If nonzero $V$ is not of rank one, then Lemma 2.3 of \cite{marcin_exp} tells us that the pair $(V^\perp, \{0\}^\perp)$
satisfies the range criterion. This is equivalent to say that both $\phi^V$ and $\phi_V$ generate rays
in the cone of $\mathbb D$ which are exposed by separable states by \cite{kye_decom}.
If $V=xy^*$ is of rank one, then we see that $\phi_{\bar xy^*}=\phi^{xy^*}$ is both completely positive
and completely copositive. In any cases, we conclude that if $\phi$ is an exposed decomposable positive map then
both $P_\phi$ and the partial conjugates of $P_\phi$ do not span the whole space.
Conversely, if both $P_\phi$ and the partial conjugates of $P_\phi$ do not span the whole space, then
the double dual ${L_\phi}^{\prime\prime}$ contains a completely positive or a completely copositive map, and so
$\phi$ is not exposed indecomposable. Therefore, we have shown that an exposed positive map $\phi$ is decomposable
if and only if both $P_\phi$ and the partial conjugates of $P_\phi$ do not span the whole space.

We recall that the interior of the face $\tau(D,E)$ is given by
$$
\inte\tau(D,E)=\{A\in\mathbb T: {\mathcal R}A= D,\ {\mathcal R}A^\tau= E\}.
$$
Especially, the interior of the $\mathbb T$ itself consists of all $A\in\mathbb T$ such that both $A$ and $A^\tau$
have the full ranges. Therefore, for a positive map $\phi$, we see that
the dual face ${L_\phi}^\prime$ lies in the interior of $\mathbb T$
if and only if both $P_\phi$ and the partial conjugates of $P_\phi$ span the whole space.
We summarize as follows:

\begin{proposition}
Let $\phi$ be a positive linear map. Then the dual face ${L_\phi}^\prime$ lies in the interior of the cone $\mathbb T$ if and only if
both $P_\phi$ and the partial conjugates of $P_\phi$ span the whole space. If $\phi$ is exposed then this is the case if
and only if $\phi$ is indecomposable.
\end{proposition}

For $\phi\in\mathbb P_1$, consider the set $E^{\mathbb T}_\phi$ of all entangled states
with positive partial transposes which are detected by
the positive map $\phi$:
$$
E^{\mathbb T}_\phi=\{A\in\mathbb T: \lan A,\phi \ran<0\}.
$$
Note that $E^{\mathbb T}_\phi$ is nonempty if and only if $\phi$ is indecomposable.
Let $A$ be an interior point of the cone $\mathbb T$, which contains the identity matrix $I$ as a typical
interior point. If we take a line segment from $I$ to the boundary point of $\mathbb T$ through $A$, then any point
in the interior of this line segment is an interior point of $\mathbb T$. This shows the implication (i) $\implies$ (ii)
of the followings:

\begin{theorem}\label{exposed}
For a positive linear map $\phi$, the followings are equivalent:
\begin{enumerate}
\item[(i)]
Both product vectors in $P_\phi$ and their partial conjugates span the whole space.
\item[(ii)]
the set $E^{\mathbb T}_\phi$ has contains an entangled state $A$ such that both $A$ and $A^\tau$ have trivial kernels.
\item[(iii)]
The set $E^{\mathbb T}_\phi$ has a nonempty relative interior in $\mathbb T$.
\end{enumerate}
If $\phi$ is an exposed indecomposable positive map then the above these conditions hold.
\end{theorem}

\begin{proof}
It remains to prove the implication {\rm (ii)} $\Longrightarrow$ {\rm (i)}. Suppose that
both $A\in E^{\mathbb T}_\phi$ and $A^\tau$ have the full ranges, and consider the line segment
between $A$ and the identity matrix $I$ in $M_n\otimes M_m$. Since $\lan A,\phi\ran <0$ and $\lan I,\phi\ran >0$,
there is $A_0$ on the line segment such that $\lan A_0,\phi\ran =0$.
Denote by $D$ and $E$ the orthogonal complements of the product vectors in $P_\phi$ and the partial conjugates of product vectors
in $P_\phi$, respectively. Then we see that $A_0$ belongs to the face $\sigma(D,E)^\prime$ of $\mathbb T$.
Since $A_0$ is an interior point of $\mathbb T$, we conclude that both $D$ and $E$ are zeroes.
\end{proof}

Recall the generalized Choi maps between $M_3$ defined by
\begin{equation}\label{choi}
\Phi[a,b,c](X)=\\
\begin{pmatrix}
ax_{11}+bx_{22}+cx_{33} & -x_{12} & -x_{13} \\
-x_{21} & cx_{11}+ax_{22}+bx_{33} & -x_{23} \\
-x_{31} & -x_{32} & bx_{11}+cx_{22}+ax_{33}
\end{pmatrix}
\end{equation}
for $X=[x_{ij}]\in M_3$, as was introduced in \cite{cho-kye-lee}, where $a,b$ and $c$ are nonnegative numbers.
Note that the set $P_{\Phi[1,0,1]}$ for the Choi map $\Phi[1,0,1]$ spans $7$-dimensional subspace and their partial conjugates
span the whole space $\mathbb C^3\otimes \mathbb C^3$, as was shown in \cite{choi_kye}. This means that
the set $E^{\mathbb T}_{\Phi[1,0,1]}$ of PPT entangled states detected by the Choi map $\Phi[1,0,1]$ lies on the boundary of the cone
$\mathbb T$. Even if it happens that the dual face $F$ of the Choi map
in $\mathbb V_1$ is on the boundary, it should be noted
that the smallest face $F_1$ of $\mathbb T$ containing $F$ is much bigger than $F$ itself, as was shown in \cite{choi_kye}. The face
$F$ is a face of the convex cone $\mathbb V_1\cap F_1$, which is a part of the boundary of $\mathbb V_1\cap F_1$.

Especially, we see that the Choi map is not exposed as was already noticed in  \cite{eom-kye} and \cite{kye-canad},
even though it generates an extremal ray \cite{choi-lam}.
In the next section, we will show that $\Phi[a,b,c]$ is an exposed positive linear map whenever the following conditions
\begin{equation}\label{cond}
0< a< 1,\qquad a+b+c= 2,\qquad bc= (1-a)^2
\end{equation}
hold, which gives us a one parameter family of exposed positive maps.

\section{exposed indecomposable positive maps}

In spite of their usefulness, the whole structures of the convex cone $\mathbb P_1$
is far from being completely understood even for the cases when $m$ and $n$ are small numbers.
For the case of $m=m=2$, all extreme points of the convex set consisting of unital positive linear maps
were found by St\o rmer \cite{stormer} in 1963. Furthermore, all faces of the cone $\mathbb P_1$ were characterized in
\cite{byeon-kye} in terms of certain pairs of subspaces of $M_2$. See also \cite{kye-2by2_II} for the faces
of the convex set of all unital positive linear maps in $M_2$.

The first example of a map of the type (\ref{choi}) was given by Choi \cite{choi72},
who showed that the map $\Phi[1,2,2]$ is a $2$-positive linear map which is not completely positive.
This is the first example to distinguish $n$-positivity for different $n$'s.
The first examples of indecomposable positive linear maps by Choi mentioned in Introduction are $\Phi[1,0,\mu]$ for
$\mu\ge 1$. See also \cite{stormer82} for another method to prove that.
The map $\Phi[1,0,1]$, which is usually called the Choi map,
was shown \cite{choi-lam} to  generate an extremal ray of the cone $\mathbb P_1$.
Furthermore, it turns out \cite{tomiyama} that this map $\Phi[1,0,1]$ is an atom, that is,
it is not the sum of a $2$-positive map and a $2$-copositive map.

There had been very few examples of indecomposable positive linear maps in the literature until
it was shown in \cite{cho-kye-lee} that $\Phi[a,b,c]$ is positive if and only if
$$
a+b+c\ge 2,\qquad 0\le a\le 1\implies bc\ge (1-a)^2,
$$
and decomposable if and only if
$$
0\le a\le 2\implies bc\ge \left(\frac{2-a}2 \right)^2.
$$
It was also shown that $\Phi[a,b,c]$ is completely positive if and only if $a\ge 2$ and it is completely copositive
if and only if $bc\ge 1$.
For more extensive examples of indecomposable positive linear maps, we refer to \cite{cw-indec}.
We note that there are another variant of the Choi map as was considered in \cite{blau}. Some of them,
parameterized by three real variables, were shown \cite{osaka} to be extremal.
See also \cite{ber} and \cite{ck-Choi}  for another variations of the Choi map.

The most interesting cases arise when
$$
0\le a\le 1,\qquad a+b+c= 2,\qquad bc= (1-a)^2,
$$
as was analyzed in \cite{cw}. See also \cite{koss} and \cite{ck}  for the different approach and generalization of these maps. Note that
the maps $\Phi[1,0,1]$ and $\Phi[1,1,0]$ reproduce the Choi map and its dual,
respectively, in the case of $a=1$. On the other hand, if $a=0$ then the map $\Phi[0,1,1]$ is completely
copositive. Answering a question raised in \cite{cw}, the authors \cite{ha+kye_indec-witness} have shown that
the map $\Phi[a,b,c]$ with the condition (\ref{cond})
gives rise to an indecomposable optimal witness. After the authors circulated and posted the paper \cite{ha+kye_indec-witness},
Chru\'{s}ci\'{n}ski \cite{cw_letter} kindly sent the authors their proof which is independent from the authors.

In this section, we show that the map $\Phi[a,b,c]$ generates an exposed ray
of the cone $\mathbb P_1$ under the condition (\ref{cond}).
For those maps,
the authors \cite{ha+kye_indec-witness} could find product vectors in $P_{\Phi[a,b,c]}$ both of whom and whose partial conjugates
span the whole space $\mathbb C^3\otimes \mathbb C^3$. The point was to parameterize them by
$$
a(t)=\dfrac{(1-t)^2}{1-t+t^2},\quad b(t)=\dfrac{t^2}{1-t+t^2},\quad c(t)=\dfrac 1{1-t+t^2}.
$$
Note that
$$
0\le a(t)\le 1,\quad a(t)+b(t)+c(t)=2,\quad b(t)c(t)=(1-a(t))^2.
$$
We denote by $\Phi(t):=\Phi[a(t),b(t),c(t)]$.
Note that $t=1$ corresponds the completely copositive reduction map $\Phi(1)=\Phi[0,1,1]$, and $t=0$ corresponds to
the Choi map $\Phi(0)=\Phi[1,0,1]$.
With this parameterization, we found product vectors $\bar y_i\otimes x_i$ in $P_{\Phi(t)}=L_{\Phi (t)}'$ as follows:
$$
\begin{aligned}
x_1 &=(1,1,1)^{\rm t},\\
\bar y_1 &=(1,1,1)^{\rm t},\\
\end{aligned}
\quad
\begin{aligned}
x_2 &=(1,-1,1)^{\rm t},\\
\bar y_2 &=(1,-1,1)^{\rm t},\\
\end{aligned}
\quad
\begin{aligned}
x_3 &=(1,i,-i)^{\rm t},\\
\bar y_3 &=(1,-i,i)^{\rm t},\\
\end{aligned}
$$
$$
\begin{aligned}
x_4 &=(0,\sqrt t,1)^{\rm t},\\
x_6 &=(1,0,\sqrt t)^{\rm t},\\
x_8 &=(\sqrt t,1,0)^{\rm t},\\
\end{aligned}
\quad
\begin{aligned}
x_5 &=(0,\sqrt t,i)^{\rm t},\\
x_7 &=(i,0,\sqrt t)^{\rm t},\\
x_9 &=(\sqrt t, i,0)^{\rm t},\\
\end{aligned}
\quad
\begin{aligned}
\bar y_4 &=(0,\sqrt t,t)^{\rm t},\\
\bar y_6 &=(t,0,\sqrt t)^{\rm t},\\
\bar y_8 &=(\sqrt t,t,0)^{\rm t},\\
\end{aligned}
\quad
\begin{aligned}
\bar y_5 &=(0,\sqrt t,-ti)^{\rm t},\\
\bar y_7 &=(-ti,0,\sqrt t)^{\rm t},\\
\bar y_9 &=(\sqrt t, -ti,0)^{\rm t}.\\
\end{aligned}
$$
Note \cite{ha+kye_indec-witness} that both these product vectors
$\bar y_i \otimes x_i$ and their partial complex conjugates
$\bar y_i\otimes \bar x_i$ span the whole space $\mathbb C^3\otimes \mathbb C^3$ for any positive number $t$ except for $t=1$.
Actually, all (unnormalized) product vectors $\bar y\otimes x$ in $P_{\Phi(t)}$ are of the form
\begin{equation}\label{eq:case1}
(\bar a_1,\bar a_2,\bar a_3)^{\rm t}\otimes (a_1,a_2,a_3)^{\rm t}\text{ with }|a_1|=|a_2|=|a_3|,
\end{equation}
\begin{equation}\label{eq:case2}
(0,\bar a_2, t\bar a_3)^{\rm t}\otimes (0,a_2,a_3)^{\rm t} \text{ with }|a_2|^2=t|a_3|^2,
\end{equation}
\begin{equation}\label{eq:case3}
(t\bar a_1,0, \bar a_3)^{\rm t}\otimes (a_1,0,a_3)^{\rm t} \text{ with }|a_3|^2=t|a_1|^2,
\end{equation}
\begin{equation}\label{eq:case4}
(\bar a_1,t\bar a_2, 0)^{\rm t}\otimes (a_1,a_2,0)^{\rm t} \text{ with }|a_1|^2=t|a_2|^2.
\end{equation}

Now, we consider the $9\times 9$ matrices which represent the positive linear maps in the double dual
of the map $\Phi(t)$ under the Jamio\l kowski-Choi isomorphism. We note that the vectors in \eqref{eq:case2}
determine the $4\times 4$ submatrices taking the entries from the $5,6,8$ and $9$-th rows and columns, which represent
a positive linear maps from $M_2$ into $M_2$. We also recall that every positive maps between $M_2$ is decomposable,
and the orthogonal matrices to product vectors in \eqref{eq:case2} and their partial conjugates are
$$
A=\left(\begin{matrix}1&0\\0&-1\end{matrix}\right),\qquad
B=\left(\begin{matrix}0&1\\-t&0\end{matrix}\right),
$$
respectively. Note that the map $\alpha  \phi_A+q \phi^B$ is represented by the matrix
$$
\alpha  \phi_A+q \phi^B= \left(\begin{matrix}
\alpha  &0&0&-\alpha  -tq \\0&q &0&0\\0&0&t^2q &0\\-\alpha-tq &0&0&\alpha  \end{matrix} \right).
$$
Considering the other cases of \eqref{eq:case3} and \eqref{eq:case4} similarly,
we see that the matrices representing positive maps in
the double dual face ${L_{\Phi(t)}}''$ are of the form
$$
\left(
\begin{matrix}
\alpha  &\cdot&\cdot    &    \cdot & -\alpha  -tp & \cdot    &    \cdot &\cdot & -\alpha  -tr \\
\cdot&p &\cdot    &    \cdot & \cdot & \cdot    &    \cdot &\cdot & \cdot\\
\cdot&\cdot&t^2r    &    \cdot & \cdot & \cdot    &    \cdot &\cdot & \cdot\\
\cdot&\cdot&\cdot    &  t^2p & \cdot & \cdot    &    \cdot &\cdot & \cdot\\
-\alpha  -tp &\cdot&\cdot    &    \cdot & \alpha  & \cdot    &    \cdot &\cdot & -\alpha  -tq \\
\cdot&\cdot&\cdot    &    \cdot & \cdot & q    &    \cdot &\cdot & \cdot\\
\cdot&\cdot&\cdot    &    \cdot & \cdot & \cdot    &    r &\cdot & \cdot\\
\cdot&\cdot&\cdot    &    \cdot & \cdot & \cdot    &    \cdot &t^2q & \cdot\\
-\alpha  -tr &\cdot&\cdot    &    \cdot & -\alpha  -tq & \cdot    &    \cdot &\cdot & \alpha  \end{matrix}
\right),
$$
for nonnegative numbers $\alpha, p,\,q,\,r$.
Considering the product vectors \eqref{eq:case1}, we have also the condition
\begin{equation}\label{eq:pos_cond}
3\alpha = (1-t)^2 (p+q+r).
\end{equation}
Note that the corresponding maps send $[x_{ij}]$ to
$$
\left(\begin{matrix}
\alpha x_{11}+t^2px_{22}+rx_{33} & (-\alpha -tp)x_{12}  & (-\alpha -tr)x_{13}\\
(-\alpha -tp)x_{21} & \alpha x_{22}+t^2qx_{33}+px_{11}  & (-\alpha -tq)x_{23}\\
(-\alpha -tr)x_{31} & (-\alpha -tq)x_{32} &  \alpha x_{33}+t^2rx_{11}+qx_{22}
\end{matrix}\right)
$$
Therefore, this map is positive if and only if
\begin{equation}\label{eq:iff}
\left(\begin{matrix}
\alpha x^2+4p y^2+r z^2 & (-\alpha -2p)xy  & (-\alpha -2r)xz\\
(-\alpha -2p)yx & \alpha y^2+4qz^2+px^2  & (-\alpha -2q)yz\\
(-\alpha -2r)zx & (-\alpha -2q)zy &  \alpha z^2+4rx^2+qy^2
\end{matrix}\right)
\end{equation}
is positive semi-definite for every real numbers $x,y,z$.

Define $S_1,\,S_2,\,S_3$ by
\[
S_1=p+q+r,\quad S_2=pq+qr+rp,\quad S_3=pqr.
\]
By considering the determinant of the matrix \eqref{eq:iff} with $(x,y,z)=(1,1,1)$,  we have the following necessary condition for the positivity of the maps
\[
D[\alpha,p,q,r]=S_3t^6+\alpha S_2t^4-2(S_3+\alpha S_2)t^3-\alpha S_2 t^2
-2\alpha(S_2+2\alpha S_1)t +S_3 +\alpha S_2 -4\alpha^3\ge 0.
\]
On the other hand, since $\alpha=\frac 13 (1-t)^2 S_1$ from \eqref{eq:pos_cond}, we can write
\begin{align*}
D[\alpha,p,q,r]=&(t^6-2t^3+1)T_1-(t^5-t^4-t^2+t)T_2\\
                    =&(t-1)^2(t^2+t+1)(t^2+1)T_1+(t-1)^2(t^2+t+1)t(T_1-T_2),
\end{align*}
where
\begin{align*}
T_1=&-\frac 4{27}S_1^3+\frac 13 S_1 S_2+S_3,\\
T_2=&-\frac 49 S_1^3 + \frac 43 S_1 S_2.
\end{align*}
Now, by the inequality of arithmetic and geometric means, we have
\begin{align*}
 T_1=&\frac {10}9 pqr-\frac{4}{27}(p^3+q^3+r^3)
     -\frac 19 (p^2 q+p^2 r +pq^2+pr^2+q^2r+qr^2)\\
     \le & \frac 49 pqr -\frac{4}{27}(p^3+q^3+r^3).
 \end{align*}
Denote $T_3$ by the right hand side of the above inequality, then
\begin{align*}
 T_1-T_2=&\frac{8}{27}(p^3+q^3+r^3)-\frac 19(p^2 q+p^2 r +pq^2+pr^2+q^2r+qr^2)-\frac29 pqr     \\
     \le & \frac 8{27}(p^3+q^3+r^3)-\frac 89 pqr=-2T_3,
\end{align*}
Therefore, we have
\begin{align*}
D[\alpha,p,q,r]&\le (t-1)^2(t^2+t+1)(t^2+1)T_3+(t-1)^2(t^2+t+1)t(-2T_3)\\
                    &=(t-1)^4(t^2+t+1)T_3\le 0,
\end{align*}
by the inequality of arithmetic and geometric means again.

By combining this inequality with the necessary condition for positivity, we conclude that $D[\alpha,p,q,r]=0$.
This is the case only when $p=q=r$ and $\alpha=(1-t)^2 p$ by
the equality conditions for the inequalities between arithmetic and geometric means.

Consequently, the corresponding positive maps are the scalar multiples of the generalized Choi map
$$
{(1-t+t^2)}p\,\Phi[a(t),b(t),c(t)].
$$
This means that ${L_{\Phi(t)}}''=L_{\Phi(t)}$, that is, every $\Phi(t)$
generates an exposed extremal ray for any positive number $t$ except for $t=1$.

Note that an indecomposable positive linear map
$\Phi[a,b,c]$ with $a+b+c=2$ can be written as the following  convex combination
$$
(1-\alpha) \Phi[2,0,0]+\alpha\Phi(t)
$$
where $t=\sqrt{\frac bc}$ and $\alpha = c(1-t+t^2)$. Since $\Phi[2,0,0]$
is completely positive linear map, the only  optimal entanglement witnesses
arising from these maps  are witnesses arising from the exposed indecomposable positive maps $\Phi(t)$.

There are bunch of examples of optimal decomposable entanglement witness which is not
extremal. See \cite{asl} for examples. However,
all the known examples of indecomposable optimal entanglement witnesses
arise from positive linear maps which generate extremal rays. Therefore,
it would be interesting to know if there exists an example of an optimal indecomposable
entanglement witness arising from a positive linear map which is not extremal.
Finally, it would be also interesting to determine if every positive map satisfying the
conditions in Theorem \ref{exposed} is exposed or not.


\begin{thebibliography}{99}

\bibitem{arveson}
W. B. Arveson,
\it Subalgebras of $C^*$-algebras,
\rm Acta Math. \bf 123 \rm (1969), 141--224.

\bibitem{asl}
R. Augusiak, G. Sarbicki and M. Lewenstein,
\it Optimal decomposable witnesses revisited,
\rm preprint, arXiv:1107.0505.

\bibitem{ber}
R. A. Bertlmann, K. Durstberger, B. C. Hiesmayr and P. Krammer,
\it Optimal Entanglement Witnesses for Qubits and Qutrits,
\rm Phys. Rev. A \bf 72 \rm (2005), 052331.

\bibitem{byeon-kye} E.-S. Byeon and S.-H. Kye, \it Facial
structures for positive linear maps in the two dimensional matrix
algebra, \rm Positivity, \bf 6 \rm (2002), 369--380.

\bibitem{cho-kye-lee} S.-J. Cho, S.-H. Kye and S. G. Lee, \it Generalized Choi maps in
3-dimensional matrix algebras, \rm Linear Alg. Appl. \bf171 \rm (1992), 213--224.

\bibitem{choi_kye}
H.-S. Choi and S.-H. Kye,
\it Facial Structures for Separable States,
\rm J. Korean Math. Soc., to appear, \texttt{http://www.math.snu.ac.kr/{$\sim$}kye/paper/separable.pdf}.

\bibitem{choi72}
M.-D. Choi,
\it positive linear maps on $C^*$-algebras,
\rm Canad. Math. J. \bf 24 \rm (1972), 520--529.

\bibitem{choi75-10}  M.-D. Choi, \it Completely positive linear
maps on complex matrices, \rm Linear Alg. Appl. \bf 10 \rm
(1975), 285--290.

\bibitem{choi75}  M.-D. Choi, \it Positive semidefinite
biquadratic forms, \rm Linear Alg. Appl. \bf 12 \rm (1975),  95--100.

\bibitem{choi-ppt}
M.-D. Choi,
\it Positive linear maps,
\rm Operator Algebras and Applications (Kingston, 1980), pp. 583--590,
Proc. Sympos. Pure Math. Vol 38. Part 2, Amer. Math. Soc., 1982.

\bibitem{choi-lam}  M.-D. Choi and T.-T. Lam, \it Extremal positive
semidefinite forms, \rm Math. Ann. \bf 231 \rm (1977),
1--18.

\bibitem{cw-indec}
D. Chru\'{s}ci\'{n}ski and A. Kossakowski,
\it On the structure of entanglement witnesses and new class of positive indecomposable maps,
\rm Open Sys. Information Dyn. \bf 14 \rm (2007), 275--294.

\bibitem{ck-Choi}
D. Chru\'{s}ci\'{n}ski and A. Kossakowski,
\it How to construct indecomposable entanglement witnesses,
\rm J. Phys. A: Math. Theor. \bf 41 \rm (2008), 145301.

\bibitem{ck}
D. Chru\'{s}ci\'{n}ski and A. Kossakowski,
\it Geometry of quantum states: New construction of positive maps,
\rm Phys. Lett. A \bf 373 \rm (2009), 2301--2305.


\bibitem{cw}
D. Chru\'{s}ci\'{n}ski and F. A. Wudarski,
\it Geometry of entanglement witnesses for two qutrits
\rm preprint, arXiv:1105.4821.

\bibitem{cw_letter}
D. Chru\'{s}ci\'{n}ski and G. Sarbicki,
private communication.

\bibitem{eom-kye} M.-H. Eom and S.-H. Kye, \it Duality for positive linear maps in matrix
algebras, \rm Math. Scand. \bf 86 \rm (2000), 130--142.

\bibitem{gan}
N. Ganguly, S. Adhikari and A. S. Majumdar,
\it Common entanglement witnesses and their characteristics,
\rm arXiv:1101.0477.

\bibitem{ha_kye_04}
K.-C. Ha and S.-H. Kye,
\it Construction of entangled states with positive partial transposes based on indecomposable positive linear maps,
\rm Phys. Lett. A \bf 325 \rm (2004), 315--323.

\bibitem{ha+kye_indec-witness}
K.-C. Ha and S.-H. Kye,
\it One parameter family of indecomposable optimal entanglement witnesses arising from generalized Choi maps,
\rm preprint, arXiv:1107.2720.

\bibitem{horo-1}
M. Horodecki, P. Horodecki and R. Horodecki,
\it Separability of mixed states: necessary and sufficient conditions,
\rm Phys. Lett. A
\bf 223
\rm (1996),
1--8.

\bibitem{hou}
J. Hou, and Y. Guo,
\it When different entanglement witnesses detect the same entangled states
\rm Phys. Rev. A \bf 82 \rm (2010), 052301.

\bibitem{koss}
A. Kossakowski, \it A class of linear positive maps in matrix algebras,
\rm Open Sys. Information Dyn. \bf 10 \rm (2003), 213--220.

\bibitem{blau}
S.-H. Kye,
A class of atomic positive linear maps in $3$-dimensional matrix algebras,
Elementary Operator and Applicatins, Proc. Workshop on Elementary Operators (Blaubeuren, June 1991)
World Scientific, 1992, pp. 205-209.

\bibitem{kye-canad}
S.-H. Kye,
\it Facial structures for positive linear maps between matrix algebras,
\rm Canad. Math. Bull. \bf 39 \rm (1996), 74--82.


\bibitem{kye-2by2_II}
S.-H. Kye,
\it Facial structures for unital positive Linear Maps in the two dimensional matrix algebra,
\rm Linear Alg, Appl. \bf 362 \rm (2003), 57-- 73.

\bibitem{kye_decom}
S.-H. Kye,
\it Facial structures for decomplsable positive linear maps in matrix algebras,
\rm Positivity \bf 9 \rm (2005), 63--79.

\bibitem{lew00}
M. Lewenstein, B. Kraus, J. I. Cirac and P. Horodecki,
\it Optimization of entanglement witness,
\rm Phys. Rev. A \bf 62 \rm (2000), 052310.

\bibitem{lew00_1}
M. Lewenstein, B. Kraus, P. Horodecki and J. I. Cirac,
\it Characterization of separable states and entanglement witnesses,
\rm Phys. Rev. A \bf 63 \rm (2000), 044304.

\bibitem{marcin_exp}
M. Marciniak,
\it Rank properties of exposed positive maps,
\rm preprint, arXiv:1103.3497.

\bibitem{osaka}
H. Osaka,
A class of extremal positive maps in $3\times 3$ matrix algebras,
Publ. Res. Inst. Math. Sci. \bf 28 \rm  (1992), 747--756.

\bibitem{peres}
A. Peres,
\it Separability Criterion for Density Matrices,
\rm Phys. Rev. Lett. \bf 77 \rm (1996), 1413--1415.

\bibitem{rock}
R. T. Rockafellar,
Convex Analysis, Princeton University Press, 1970.

\bibitem{stine}
W. F. Stinespring,
\it Positive functions on $C^*$C-algebras,
\rm Proc. Amer. Math. Soc. \bf 6, \rm (1955), 211--216.

\bibitem{stormer}
E. St\o rmer,
\it Positive linear maps of operator algebras,
\rm Acta Math. \bf 110 \rm (1963), 233--278.

\bibitem{stormer82}  E. St\o rmer, \it Decomposable positive maps on
$C^*$-algebras, \rm Proc. Amer. Math. Soc. \bf 86 \rm (1982),
402--404.

\bibitem{stormer-dual}
E. St\o rmer,
\it Extension of positive
maps into $B(\mathcal H)$, \rm J. Funct. Anal. \bf 66 \rm (1986), 235-254.

\bibitem{tomiyama} K. Tanahashi and J. Tomiyama, \it Indecomposable
positive maps in matrix algebras, \rm Canad. Math. Bull. \bf 31 \rm
(1988), 308--317.

\bibitem{tang}
W.-S. Tang, \it On positive linear maps between matrix algebras,
\rm Linear Algebra Appl. \bf 79 \rm (1986), 33--44.

\bibitem{terhal}
B. M. Terhal,
\it Bell Inequalities and the Separability Criterion,
\rm Phys. Lett. A \bf 271 \rm (2000), 319--326.

\bibitem{woronowicz}
S. L. Woronowicz,
\it Positive maps of low dimensional matrix algebras, \rm Rep. Math. Phys. \bf 10 \rm (1976),
165--183.

\bibitem{woronowicz-1}
S. L. Woronowicz, \it  Nonextendible positive maps, \rm Comm. Math. Phys. \bf 51 \rm
(1976), 243--282.

\bibitem{woro_letter}
S. L. Woronowicz, private communication.

\bibitem{yopp}
D. A. Yopp and R. D. Hill,
\it Extremals and exposed faces of the cone of positive maps,
\rm Linear and Multilinear Algebra, \bf 53 \rm (2005), 167--174.

\end{thebibliography}
\end{document}